\documentclass[a4paper,10pt,twoside,reqno,nonamelimits,  envcountsame]{llncs}
\usepackage{makeidx}  
\usepackage{hyperref}

\usepackage{amsmath}
\usepackage{amssymb}
\usepackage{amsfonts}
\usepackage{amscd}
\usepackage{amsxtra}
\usepackage{mathabx}
\usepackage{mathrsfs}
\usepackage{nicefrac}
\usepackage{graphicx}
\usepackage{xcolor}
\usepackage{tikz}
\usepackage{wrapfig}
\usepackage{enumitem}
\setlist[enumerate]{leftmargin=2em,itemindent=0em, labelindent=0pt,labelwidth=1.5em,labelsep=.5em, align=left, noitemsep}
\newlist{txtenum}{enumerate}{1}
\setlist[txtenum]{leftmargin=0em,itemindent=1.5em, labelindent=0pt,labelwidth=1em,labelsep=.5em, align=left}

\newlength{\hspaceforlengthglumpf}

\newcommand{\One}{\mathbf{1}}

\newcommand{\lt}{\left}
\newcommand{\rt}{\right}

\newcommand{\abs}[1]{{\lt\lvert{#1}\rt\rvert}}

\newcommand{\nfrac}[2]{{\nicefrac{#1}{#2}}}

\newcommand{\ZZ}{\mathbb{Z}}

\DeclareMathOperator*{\Prb}{\mathbb{P}}
\DeclareMathOperator*{\Exp}{\mathbb{E}}

\newcommand{\eps}{\varepsilon}

\newlength{\algotabbingwidth}
\newtheorem{fact}[theorem]{Fact}

\newcommand{\todo}[1]{}

\newcommand{\sq}{\square}
\newcommand{\I}[1]{\One_{#1}}

\newcommand{\CmnE}{E^\cap}
\newcommand{\CmnEa}{E^{\cap\ast}}          \newcommand{\Sa}{S^{\ast}}
\newcommand{\RstE}{E^r}                    \newcommand{\Rst}{R}

\begin{document}
\mainmatter              
\title{The Graph of the Pedigree Polytope is Asymptotically Almost Complete (Extended Abstract)}
\titlerunning{The Graph of the Pedigree Polytope}
\toctitle{The Graph of the Pedigree Polytope is Asymptotically Almost Complete}
\author{Abdullah Makkeh\inst{1} \and Mozhgan Pourmoradnasseri\inst{1} \and Dirk Oliver Theis\inst{1}\thanks{Supported by the Estonian Research Council, ETAG (\textit{Eesti Teadusagentuur}), through PUT Exploratory Grant \#620, and by the European Regional Development Fund through the Estonian Center of Excellence in Computer Science, EXCS.}}
\authorrunning{A.Makkeh, M.Pourmoradnasseri, D.O.Theis}
\tocauthor{Abdullah Makkeh, Mozhgan Pourmoradnasseri, Dirk Oliver Theis}
\institute{%
  University of Tartu\\
  Institute of Computer Science\\
  {\"U}likooli 17\\
  51014 Tartu, Estonia,\\
  \email{\{mozhgan,dotheis\}@ut.ee}\\
  WWW: \texttt{http://ac.cs.ut.ee/}%
}
\maketitle              
\pagestyle{headings}


\begin{abstract}
  Graphs (1-skeletons) of Traveling-Salesman-related polytopes have attracted a lot of attention.
  Pedigree polytopes are extensions of the classical Symmetric Traveling Salesman Problem polytopes (Arthanari 2000) whose graphs contain the TSP polytope graphs as spanning subgraphs (Arthanari 2013).
  Unlike TSP polytopes, Pedigree polytopes are not ``symmetric'', e.g., their graphs are not vertex transitive, not even regular.

  We show that in the graph of the pedigree polytope, the quotient minimum degree over number of vertices tends to~1 as the number of cities tends to infinity.

  \medskip%
\end{abstract}

\section{Introduction}\label{sec:intro}
Steinitz's Theorem states that 3-connected planar graphs are precisely the graphs of 3-dimensional polytopes.

Properties of graphs of polytopes of higher dimension are of interest not only in the combinatorial study of polytopes, but also in Combinatorial Optimization, and Theoretical Computer Science.

The famous Hirsch conjecture in the combinatorial study of polytopes, settled by Santos~\cite{Santos:hirsch:2012}, concerned the diameter of graphs of polytopes.

In Combinatorial Optimization, the study of the graphs of polytopes associated with combinatorial optimization problems was initially motivated by the search for algorithms for these problems.

In Theoretical Computer Science, the theorem by Papadimitriou~\cite{Papadimitriou:adj-tsp:78} that Non-Adjacency of vertices of (Symmetric) Traveling Salesman Problem (TSP) polytopes is NP-complete, gave rise to similar results about other families of polytopes (cf.\ \cite{Aguilera:adj-setcover-polyh:2014,Maksimenko:adj-np:2014} and the references therein, for recent examples).

There have been particularly many attempts to understand the graph of TSP polytopes, and, where this turned out to be infeasible, of TSP-related polytopes (e.g., \cite{Sierksma:DAM:00}; cf.\ \cite{Groetschel-Padberg:polyh-th:1985,NaddefRinaldi93,OswaldReineltTheis07,TheisDAM10,Theis:DiscreteOpt:2014} and the references therein).  The presence of long cycles has been studied (\cite{Sierksma:skeleton:1993}, see also \cite{Naddef-Pulleyblank:JCTB:84,Naddef:JCTB:84}), as has the graph density / vertex degrees (e.g., \cite{Sarangarajan:tspskel-density-LB:1997}, see also \cite{Kaibel-Remshagen:03,Kaibel:lowdim-faces-rndptp:2004}).

The motivation for the research in this paper was a 1985 conjecture by Gr\"otschel and Padberg~\cite{Groetschel-Padberg:polyh-th:1985} --- well-known in polyhedral combinatorial optimization --- stating that the graph of TSP polytopes has diameter~2 (also, Problem \# 36 in \cite{Schrijver:Book:03}).  Already in~\cite{Groetschel-Padberg:polyh-th:1985}, Gr\"otschel and Padberg extend the question for the diameter to a family of TSP-related polytopes which seemed easier to understand at the time.

A more recent family of TSP-related polytopes are the \textit{Pedigree polytopes} of Arthanari~\cite{Arthanari:ped-one:2000}.  For this family of polytopes, adjacency of vertices can be decided in polynomial time~\cite{Arthanari:DiscreteMath:06}.  Moreover, the graphs of the TSP polytopes are spanning subgraphs of the graphs of the Pedigree polytopes~\cite{Arthanari:DiscreteOpt:13}.  This is so mainly because the Pedigree polytope for~$n$ cities is an \textit{extension,} without ``hidden'' vertices, of the TSP polytope (cf., \cite{Pashkovich:hidden:15,Fiorini-Kaibel-Pashkovich-Theis:CombLB:13}).

In this paper, we prove the following about graphs of Pedigree polytopes.  Recall that the number of vertices of the TSP polytope for~$n$ cities is the number of cycles with vertex set $[n] := \{1,\dots,n\}$, which is $(n-1)!/2$.

\begin{theorem}\label{thm:main-polytope}
  The minimum degree of a vertex on the Pedigree polytope for~$n$ cities is $(1-o(1))\cdot(n-1)!/2$ (for $n\to\infty$).
\end{theorem}

In particular, the density graph of Pedigree polytopes is asymptotically equal to~1.  Note, though, that while for TSP polytopes, these two statements are equivalent, this is not the case of Pedigree polytopes.  The reason is that Pedigree extensions are not as ``symmetric'' as TSP polytopes (cf.~\cite{Kaibel-Pashkovich-Theis:SIDMA:2012}): For every two vertices $u,v$ of the TSP polytope for~$n$ cities, there is an affine automorphism of the polytope mapping $u$ to~$v$.  (Similar statements are true for monotone-TSP and graphical-TSP polytopes.)  This is not true for Pedigree polytopes: Arthanari's construction removes the symmetry to a large extent.

Numerical simulations show that, even for relatively large $n$ (say, $\approx 100$), the graph of the Pedigree polytope is not complete.  We have made no attempt, however, to find a non-trivial upper bound for the minimum degree.

\paragraph{We now give a non-technical description of the proof of Theorem~\ref{thm:main-polytope}.}
Arthanari's beautiful idea of a pedigree is that of a cycle ``evolving'' over time: Starting from the unique cycle with node set $\{1,2,3\}$ at time~$3$, at time~$n\ge 4$, the node~$n$ is added to the cycle by subdividing one of its edges.  We say that $n$ is \textit{inserted into} that edge.

Arthanari's combinatorial condition for adjancency on the Pedigree polytope can be thought of as a process, too, with a \textit{pedigree graph} $G$ evolving over time.  Suppose we have two evolving cycles.  Let us refer to $A$ as Alice's cycle, and to $B$ as Bob's cycle.  At time~$n$, Alice chooses an edge of her current cycle~$A$ (with node set $[n-1]$) and inserts her new node~$n$ into that edge to form her new cycle (with node set $[n]$).  At the same time, Bob chooses an edge of his current cycle~$B$, and inserts his new node~$n$ into that edge to form his new cycle.

The pedigree graph~$G$ may also change at time~$n$.  The new pedigree graph is either equal to the current one, or arises from the current one by adding the new vertex\footnote{We speak of \textit{vertices} of the pedigree graph and \textit{nodes} of the cycles, to limit confusion.} $n$ with incident edges.
The choices of Alice and Bob determine: whether the new vertex is added or not; the number of edges incident to the new vertex~$n$;  the end vertices of these edges.

Arthanari's combinatorial characterization of adjacency on the Pedigree polytope is now this.
\begin{theorem}[\cite{Arthanari:DiscreteMath:06}]\label{thm:arthanari}
  At all times $n\ge 4$, the two vertices of the Pedigree polytope for~$n$ cities corresponding to the (new) cycles $A$ and~$B$ with node set~$[n]$ are adjacent in the Pedigree polytope, if, and only if, the (new) graph~$G$ is connected.
\end{theorem}

Theorem~\ref{thm:main-polytope} states that, if~$B$ is a cycle chosen uniformly at random from all cycles on~$[n]$, then
\begin{equation*}
  \min_{A} \Prb\bigl( \text{$\{$ the pedigree graph is connected $\}$} \bigr) = 1-o(1),
\end{equation*}
where the minimum ranges over all cycles on~$[n]$.  Lower bounding this quantity amounts to studying the following ``connectivity game'':  Alice's goal is to make the graph~$G$ disconnected; whereas Bob makes uniformly random choices all the time.  We prove that Alice loses with probability $1-o(1)$.  To analyze the game, we study a kind of a Markov Decision Process with state space $\ZZ_+\times \ZZ_+$.  The states are pairs $(s,t)$, where~$s$ is the number of common edges in Alice's and Bob's cycles, and~$t$ is the number of connected components of the current pedigree graph.

\paragraph{In the next section,}
we will give rigorous statements corresponding to the handwaving explanations above.  In Section~\ref{Sec:rnd-cyc-basics}, we prove some basic facts about Bob playing randomly, and discuss the intuition of the proof of the main theorem.  In Section~\ref{Sec:cnct-game}, we introduce the Markov-Decision-Problem-ish situation that Alice finds herself in.  The proof of the main theorem is sketched in Section~\ref{Sec:mainproof}; due to space limitations we have to refer to the upcoming journal version~\cite{Makkeh-Pourmoradnasseri-Theis:pedigree:18} of the paper for the details.  We conclude with a couple of questions for future research which we find compelling.

\section{Exact Statements of the Definitions, Facts, and Results}\label{sec:exact-statements}
\subsection{Cycles, One Node at a Time}
Our cycles are undirected (so, e.g., there is only one cycle on 3 nodes).
For ease of notation, let us say that the \textit{positive direction} on a cycle with node set~$[n]$, $n\ge 3$, is the one in which, when starting from the node~1, the node~2 comes before the node~3; the other direction the \textit{negative direction.}
When referring to the $k$th edge of a cycle, we count the edges in the positive direction; the 1st one being the one incident on node~1.  E.g., in the unique cycle with node set $\{1,2,3\}$, the 1st edge is $\{1,2\}$, the 2nd edge is $\{2,3\}$, and the 3rd edge is $\{3,1\}$.

As mentioned in the introduction, Arthanari's Pedigree is a combinatorial object representing the ``evolution'' of a cycle ``over time'', and the combinatorial definition of adjacency of pedigrees makes use of that step-by-step development.  The set of Pedigrees is in bijecton with the set of cycles.  In our context (we do not have to associate points in space with Pedigrees), defining Pedigrees and then explaining the bijection with cycles is more cumbersome than necessary.  For convenience, we use the following more convenient definitions, which mirror the definition of Pedigrees, but they use cycles only.  Let us say that an \textit{infinite cycle}\footnote{The reason why we use this notion of ``infinite cycle'' is pure convenience.  It does not add complexity, but it makes many of statements and proofs less cumbersome.  Indeed, instead of an infinite cycle, it is ok to just use a cycle whose length~$M$ is longer than all the lengths occuring in the particular argument.  So instead of ``let $A$ be an infinite cycle, and consider $A_k$, $A_\ell$, $A_n$'' you have to say ``let $M$ be a large enough integer, $A_M$ a cycle of length~$M$, and $A_k$, $A_\ell$, $A_n$ sub-cycles of $A_M$''.  All the little arguments (e.g., Fact~\ref{fact:Bn-uniform} below) have to be done in the same way.} %
is a sequence $A = c_{\sq} \in \prod_{n=3}^\infty [n]$.  An infinite cycle~$A$ gives rise to an infinite sequence $A_{\sq}$ of finite cycles (in the usual graph theory sense), defined inductively as follows:
\begin{itemize}
\item $A_3$ is the unique cycle with node set $\{1,2,3\}$;
\item for $n\ge 3$, $A_n$ is the cycle with node set $[n]$ which arises from adding the node $n$ to $A_{n-1}$ by inserting it into (i.e., subdividing) the $c_{n-1}$th edge.
\end{itemize}
We think of $A_{\sq}$ as a cycle developing over time: At time~$n$, the node~$n$ is added.

We will need to access the neighbors of node~$n$ in~$A_n$, i.e., the ends of the edge into which~$n$ is inserted (i.e., which is subdivided) when~$n$ is added to~$A_{\sq}$.
We write $\nu_A^+(n)$ for the neighbor of~$n$ in~$A_n$ following~$n$ in the positive direction, and $\nu_A^-(n)$ for the neighbor of~$n$ in~$A_n$ following~$n$ in the negative direction.
The unordered pair $\nu_A(n) = \{ \nu_A^+(n), \nu_A^-(n) \}$ is the $c_{n-1}$th edge of $A_{n-1}$, the one into which~$n$ was inserted.

These definitions are for $n\ge 4$ but extend naturally for $n=1,2,3$: for $n=3$ we let $\nu_A^+(3)=1$, and $\nu_A^-(3)=2$; for $n=2$, we let $\nu_A^+(2)=\nu_A^-(2)=1$. The equation $\nu_A(n) = \{ \nu_A^+(n), \nu_A^-(n) \}$ holds for $n\ge 2$ (so $\abs{\nu_A(2)}=1$); for $n=1$ we have $\nu_A(1) := \emptyset$.

\begin{remark}[\it Finding $\nu(k)$ for ``old'' nodes~$k$]\label{rem:past-nu}
  It is readily verfied directly from the definition, that, for $k\ge 2$, $\nu^\pm_A(k)$ can be found as follows: start from node~$k$ and walk in positive direction.  The first node smaller than~$k$ which you encounter is $\nu^+_A(k)$.  Similarly, if you walk in negative direction starting from~$k$, the first node smaller than~$k$ which you hit, is $\nu^-_A(k)$.
\end{remark}

A pair of nodes $i,j$ split each cycle $A_{n}$, $n>i,j$ into two (open) segments ($i,j$ do not belong to either segment).  We say that the \textit{segment between $i$ and~$j$} is the one which does \textsl{not} contain the node $\min(\{1,2,3\}\setminus\{i,j\})$ (i.e., 1, unless $1\in\{i,j\}$, in that case, 2, unless $\{1,2\}=\{i,j\}$, in that case~3).  Note that this does not depend on the choice of $n>i,j$, which justifies to say \textit{``the segment of $A_{\sq}$ between $i$ and~$j$''.}

\begin{remark}[\it Testing/finding~$n$ with $\nu(n)=\{i,j\}$]\label{rem:check-if-pair-is-a-nu}
  Given a pair of nodes $\{i,j\}$ and $n'>i,j$, there exists an~$n\le n'$ with $\nu_A(n)=\{i,j\}$ if, and only if, the segment between $i$ and~$j$ on $A_{n'}$ is non empty and every node in it is larger than both $i$ and~$j$.
  In that case, the smallest node, $n$, in the segment between $i$ and~$j$ on $A_{\sq}$ is the one with $\nu(n)=\{i,j\}$.
\end{remark}

\subsection{The Pedigree Graph}
Two infinite cycles $A,B$ give rise to a sequence of graphs $G^{AB}_{\sq}$ which we call the \textit{pedigree graphs.}  We omit the superscripted $A,B$ when possible.  We speak of \textit{vertices} of the pedigree graphs (rather than nodes).  We do this to avoid confusion between the nodes of the cycles $A_{\sq}$,$B_{\sq}$ and the vertices of $G^{AB}_{\sq}$, because the vertex set of $G_n$ is a subset of $\{4,\dots,n\}$, and hence of the node set of $A_n$ and $B_n$.  So a node $k\in [n]$ may or may not be a vertex of $G_n$.

The pedigree graph $G_{n-1}$ is the subgraph of $G_n$ induced by the vertices in $[n-1]$.  In other words, $G_n$ is either equal to $G_{n-1}$ (if $n$ is not a vertex), or it arises from $G_{n-1}$ by adding the vertex~$n$ together with edges between~$n$ and vertices in $[n-1]$.

\begin{example}
  $G_1,G_2,G_3$ are graphs without vertices.  $G_4$ may be a graph without vertices, or it may consist of a single isolated vertex~$4$.  $G_5$ could be a graph without vertices; a graph with a single vertex~$5$; a graph with two isolated vertices $4$, $5$, or a graph with two vertices $4$, $5$, linked by an edge. Check figure \ref{fig.1} for possible $G_4$ and $G_5$.
\end{example}

According to Arthanari~\cite{Arthanari:DiscreteMath:06,Arthanari:DiscreteOpt:13} the condition for the existence of vertices is the following:

\begin{enumerate}[label=(\arabic*)]
\item A node~$n\in[n]$ is a vertex of $G_n$, iff $\nu_A(n) \ne \nu_B(n)$.
\end{enumerate}
There are several conditions for the presence of edges between the vertex~$n$ and earlier vertices.  To make it easier to distinguish these, we speak of edge ``types'' and give the edges implicit ``directions:'' from $A$ to~$B$ or from $B$ to~$A$.  Here are the conditions for edges from~$n$ to earlier vertices.
\begin{enumerate}[resume,label=(\arabic*)]
\item There is a \textit{type-1 edge ``from $A$ to~$B$''} between $n$ and $k\in [n-1]$, if $\nu_A(n) = \nu_B(k)$.  (Note that the condition implies that~$k$ is a vertex.)
\item There is a \textit{type-1 edge ``from $B$ to~$A$''} Ditto, with $A$ and~$B$ exchanged.
\item There is a \textit{type-2 edge ``from $A$ to~$B$''} between $n$ and~$\ell := \max\nu_A(n)$, unless $\displaystyle \nu_B( \ell ) \cap \nu_A(n) \neq \emptyset$.  In other words, suppose the node~$n$ was inserted into the edge $\{k,\ell\}$ in~$A$, with $k < \ell$.  Now look up the end-nodes of the edge $\nu_B(\ell)$ into which~$\ell$ was inserted when it was added to~$B$.  Unless~$k$ coincides with one of these end nodes, there is an edge between $n$ and~$\ell$.
\item \textit{Type-2 edge ``from $B$ to~$A$''} Ditto, with $A$ and~$B$ exchanged.
\end{enumerate}

Arthanari's theorem~\cite{Arthanari:DiscreteMath:06} (Theorem~\ref{thm:arthanari}) states that, if $n\ge 4$, and $A_n$, $B_n$ are two cycles with node set $[n]$, then the two vertices of the Pedigree polytope (for~$n$ cities) corresponding to $A_n$ and $B_n$ are adjacent, if, and only if, $G^{AB}_n$ is connected.

We will always think of~$A$ as ``Alice's cycle'' and $B$ as ``Bob's cycle''.

\begin{example}
  Going through an example will help understand the definition of a pedigree graph.  Figure~\ref{fig.1} shows two cycles $A$ and~$B$ evolving over time $n=3,\dots,10$, together with the evolving pedigree graph $G^{AB}_{\sq}$.

\begin{figure}[hbp]
		\begin{center}
			\begin{tikzpicture}
					\draw (0,-7) -- (0,13);
					\draw (0,-7) -- (15,-7);
					\draw (0,13) -- (15,13);
					\draw (0,12) -- (15, 12);
					\draw (15,-7) -- (15,13);
						\draw (1,-7) -- (1,13);
						\draw (6,-7) -- (6,13);
						\draw (11,-7) -- (11,13);
						\draw (0,-3.5) -- (15,-3.5);
						\draw (0,0) -- (15,0);
						\draw (0,3) -- (15,3);
						\draw (0,5.5) -- (15,5.5);
						\draw (0,8) -- (15,8);
						\draw (0,10) -- (15,10);
						\draw (0,11) -- (15,11);
					\draw (0.5,12.5) node {$n$};
					\draw (3.5,12.5) node {$A_n$};
					\draw (8.5,12.5) node {$B_n$};
					\draw (13,12.5) node {$G_n^{AB}$};
					\draw (0.5,-5.25) node {10};
					\draw (0.5,-1.75) node {9};
					\draw (0.5,1.5) node {8};
					\draw (0.5,4.25) node {7};
					\draw (0.5,6.75) node {6};
					\draw (0.5,9) node {5};
					\draw (0.5,10.5) node {4};
					\draw (0.5,11.5) node {3};
					\draw (3.5,11.5) circle (0.4);
					\draw (3.5,10.5) circle (0.4);
					\draw (3.5,9) circle (0.9);
					\draw (3.5,6.75) circle (1.1);
					\draw (3.5,4.25) circle (1.1);
					\draw (3.5,1.5) circle (1.35);
					\draw (3.5,-1.75) circle (1.6);
					\draw (3.5,-5.25) circle (1.6);
					\draw (8.5,11.5) circle (0.4);
					\draw (8.5,10.5) circle (0.4);
					\draw (8.5,9) circle (0.9);
					\draw (8.5,6.75) circle (1.1);
					\draw (8.5,4.25) circle (1.1);
					\draw (8.5,1.5) circle (1.35);
					\draw (8.5,-1.75) circle (1.6);
					\draw (8.5,-5.25) circle (1.6);
						\shade[ball color=green] (3.2, 11.8) circle (2pt);
						\shade[ball color=green] (3.2, 11.2) circle (2pt);
						\shade[ball color=green] (3.9, 11.5) circle (2pt);
						\draw (3,11.8) node {\small 1};
						\draw (3,11.2) node {\small 2};
						\draw (4.1,11.5) node {\small 3};
						\shade[ball color=green] (3.2, 10.8) circle (2pt);
						\shade[ball color=green] (3.2, 10.2) circle (2pt);
						\shade[ball color=green] (3.75, 10.2) circle (2pt);
						\shade[ball color=green] (3.75, 10.8) circle (2pt);
						\draw (3,10.8) node {\small 1};
						\draw (3,10.2) node {\small 4};
						\draw (4,10.2) node {\small 2};
						\draw (4,10.8) node {\small 3};
						\shade[ball color=green] (2.9, 9.7) circle (2pt);
						\shade[ball color=green] (2.9, 8.3) circle (2pt);
						\shade[ball color=green] (4.1, 8.3) circle (2pt);
						\shade[ball color=green] (4.4, 9) circle (2pt);
						\shade[ball color=green] (4.1, 9.7) circle (2pt);
						\draw (2.7,9.8) node {\small 1};
						\draw (2.7,8.2) node {\small 4};
						\draw (4.3,8.2) node {\small 5};
						\draw (4.6,9) node {\small 2};
						\draw (4.3,9.8) node {\small 3};
						\shade[ball color=green] (3, 7.75) circle (2pt);
						\shade[ball color=green] (2.4, 6.75) circle (2pt);
						\shade[ball color=green] (3, 5.75) circle (2pt);
						\shade[ball color=green] (4, 5.75) circle (2pt);
						\shade[ball color=green] (4.6, 6.75) circle (2pt);
						\shade[ball color=green] (4, 7.75) circle (2pt);
						\draw (2.8,7.85) node {\small 1};
						\draw (2.2,6.75) node {\small 4};
						\draw (2.8,5.65) node {\small 5};
						\draw (4.2,5.65) node {\small 2};
						\draw (4.8,6.75) node {\small 6};
						\draw (4.2,7.85) node {\small 3};
						\shade[ball color=green] (2.95, 5.2) circle (2pt);
						\shade[ball color=green] (2.4, 4.25) circle (2pt);
						\shade[ball color=green] (2.95, 3.3) circle (2pt);
						\shade[ball color=green] (4.1, 3.3) circle (2pt);
						\shade[ball color=green] (4.57, 4.5) circle (2pt);
						\shade[ball color=green] (4.57, 4) circle (2pt);
						\shade[ball color=green] (4.1, 5.2) circle (2pt);
						\draw (2.75,5.3) node {\small 1};
						\draw (2.2,4.25) node {\small 4};
						\draw (2.75,3.2) node {\small 7};
						\draw (4.35,3.2) node {\small 5};
						\draw (4.8,4) node {\small 2};
						\draw (4.8,4.5) node {\small 6};
						\draw (4.35,5.3) node {\small 3};
	
						\shade[ball color=green] (2.95, 2.75) circle (2pt);
						\shade[ball color=green] (2.25, 2) circle (2pt);
						\shade[ball color=green] (2.25, 1) circle (2pt);
						\shade[ball color=green] (2.95, 0.25) circle (2pt);
						\shade[ball color=green] (4.05, 0.25) circle (2pt);
						\shade[ball color=green] (4.75, 2) circle (2pt);
						\shade[ball color=green] (4.75, 1) circle (2pt);
						\shade[ball color=green] (4.05, 2.75) circle (2pt);
						\draw (2.7,2.8) node {\small 1};
						\draw (1.95,2) node {\small 4};
						\draw (1.95,1) node {\small 7};
						\draw (2.7,0.15) node {\small 5};
						\draw (4.3,0.15) node {\small 2};
						\draw (5,1) node {\small 6};
						\draw (5,2) node {\small 8};
						\draw (4.3,2.8) node {\small 3};
						\shade[ball color=green] (2.8, -0.3) circle (2pt);
						\shade[ball color=green] (2, -1.2) circle (2pt);
						\shade[ball color=green] (2,-2.3) circle (2pt);
						\shade[ball color=green] (2.8, -3.2) circle (2pt);
						\shade[ball color=green] (4.2, -3.2) circle (2pt);
						\shade[ball color=green] (5, -2.3) circle (2pt);
						\shade[ball color=green] (5.1, -1.75) circle (2pt);
						\shade[ball color=green] (5, -1.2) circle (2pt);
						\shade[ball color=green] (4.2, -0.3) circle (2pt);
						\draw (2.6,-0.2) node {\small 1};
						\draw (1.8,-1.2) node {\small 4};
						\draw (1.8,-2.3) node {\small 7};
						\draw (2.6,-3.3) node {\small 5};
						\draw (4.4,-3.3) node {\small 2};
						\draw (5.2,-2.3) node {\small 6};
						\draw (5.3,-1.75) node {\small 8};
						\draw (5.2,-1.2) node {\small 3};
						\draw (4.4,-0.2) node {\small 9};
						\shade[ball color=green] (2.9, -3.75) circle (2pt);
						\shade[ball color=green] (2.2, -4.3) circle (2pt);
						\shade[ball color=green] (1.9,-5.25) circle (2pt);
						\shade[ball color=green] (2.2,-6.15) circle (2pt);
						\shade[ball color=green] (2.9, -6.75) circle (2pt);
						\shade[ball color=green] (4.8, -6.15) circle (2pt);
						\shade[ball color=green] (4.1, -6.75) circle (2pt);
						\shade[ball color=green] (5.1, -5.25) circle (2pt);
						\shade[ball color=green] (4.8, -4.3) circle (2pt);
						\shade[ball color=green] (4.1, -3.75) circle (2pt);
						\draw (2.7,-3.65) node {\small 1};
						\draw (2,-4.3) node {\small 4};
						\draw (1.7,-5.25) node {\small 7};
						\draw (2,-6.15) node {\small 5};
						\draw (2.7,-6.85) node {\small 2};
						\draw (4.3,-6.85) node {\small 6};
						\draw (5,-6.15) node {\small 8};
						\draw (5.3,-5.25) node {\small 3};
						\draw (5.1,-4.3) node {\small 10};
						\draw (4.3,-3.65) node {\small 9};
						\shade[ball color=green] (8.2, 11.8) circle (2pt);
						\shade[ball color=green] (8.2, 11.2) circle (2pt);
						\shade[ball color=green] (8.9, 11.5) circle (2pt);
						\draw (8,11.8) node {\small 1};
						\draw (8,11.2) node {\small 2};
						\draw (9.1,11.5) node {\small 3};
						\shade[ball color=green] (8.2, 10.8) circle (2pt);
						\shade[ball color=green] (8.2, 10.2) circle (2pt);
						\shade[ball color=green] (8.75, 10.2) circle (2pt);
						\shade[ball color=green] (8.75, 10.8) circle (2pt);
						\draw (8,10.8) node {\small 1};
						\draw (8,10.2) node {\small 2};
						\draw (9,10.2) node {\small 3};
						\draw (9,10.8) node {\small 4};
						\shade[ball color=green] (7.9, 9.7) circle (2pt);
						\shade[ball color=green] (7.9, 8.3) circle (2pt);
						\shade[ball color=green] (9.1, 8.3) circle (2pt);
						\shade[ball color=green] (9.4, 9) circle (2pt);
						\shade[ball color=green] (9.1, 9.7) circle (2pt);
						\draw (7.7,9.8) node {\small 1};
						\draw (7.7,8.2) node {\small 5};
						\draw (9.3,8.2) node {\small 2};
						\draw (9.6,9) node {\small 3};
						\draw (9.3,9.8) node {\small 4};
						\shade[ball color=green] (8, 7.75) circle (2pt);
						\shade[ball color=green] (7.4, 6.75) circle (2pt);
						\shade[ball color=green] (8, 5.75) circle (2pt);
						\shade[ball color=green] (9, 5.75) circle (2pt);
						\shade[ball color=green] (9.6, 6.75) circle (2pt);
						\shade[ball color=green] (9, 7.75) circle (2pt);
						\draw (7.8,7.85) node {\small 1};
						\draw (7.2,6.75) node {\small 5};
						\draw (7.8,5.65) node {\small 2};
						\draw (9.2,5.65) node {\small 6};
						\draw (9.8,6.75) node {\small 3};
						\draw (9.2,7.85) node {\small 4};
						\shade[ball color=green] (7.95, 5.2) circle (2pt);
						\shade[ball color=green] (7.4, 4.25) circle (2pt);
						\shade[ball color=green] (7.95, 3.3) circle (2pt);
						\shade[ball color=green] (9.1, 3.3) circle (2pt);
						\shade[ball color=green] (9.57, 4.5) circle (2pt);
						\shade[ball color=green] (9.57, 4) circle (2pt);
						\shade[ball color=green] (9.1, 5.2) circle (2pt);
						\draw (7.75,5.3) node {\small 1};
						\draw (7.2,4.25) node {\small 5};
						\draw (7.75,3.2) node {\small 2};
						\draw (9.35,3.2) node {\small 6};
						\draw (9.8,4) node {\small 3};
						\draw (9.8,4.5) node {\small 7};
						\draw (9.35,5.3) node {\small 4};
		
						\shade[ball color=green] (7.95, 2.75) circle (2pt);
						\shade[ball color=green] (7.25, 2) circle (2pt);
						\shade[ball color=green] (7.25, 1) circle (2pt);
						\shade[ball color=green] (7.95, 0.25) circle (2pt);
						\shade[ball color=green] (9.05, 0.25) circle (2pt);
						\shade[ball color=green] (9.75, 2) circle (2pt);
						\shade[ball color=green] (9.75, 1) circle (2pt);
						\shade[ball color=green] (9.05, 2.75) circle (2pt);
						\draw (7.7,2.8) node {\small 1};
						\draw (6.95,2) node {\small 5};
						\draw (6.95,1) node {\small 2};
						\draw (7.7,0.15) node {\small 6};
						\draw (9.3,0.15) node {\small 3};
						\draw (10,1) node {\small 7};
						\draw (10,2) node {\small 4};
						\draw (9.3,2.8) node {\small 8};
						\shade[ball color=green] (7.8, -0.3) circle (2pt);
						\shade[ball color=green] (7, -1.2) circle (2pt);
						\shade[ball color=green] (7,-2.3) circle (2pt);
						\shade[ball color=green] (7.8, -3.2) circle (2pt);
						\shade[ball color=green] (9.2, -3.2) circle (2pt);
						\shade[ball color=green] (10, -2.3) circle (2pt);
						\shade[ball color=green] (10.1, -1.75) circle (2pt);
						\shade[ball color=green] (10, -1.2) circle (2pt);
						\shade[ball color=green] (9.2, -0.3) circle (2pt);
						\draw (7.6,-0.2) node {\small 1};
						\draw (6.8,-1.2) node {\small 5};
						\draw (6.8,-2.3) node {\small 2};
						\draw (7.6,-3.3) node {\small 6};
						\draw (9.4,-3.3) node {\small 3};
						\draw (10.2,-2.3) node {\small 7};
						\draw (10.3,-1.75) node {\small 4};
						\draw (10.2,-1.2) node {\small 8};
						\draw (9.4,-0.2) node {\small 9};
						\shade[ball color=green] (7.9, -3.75) circle (2pt);
						\shade[ball color=green] (7.2, -4.3) circle (2pt);
						\shade[ball color=green] (6.9,-5.25) circle (2pt);
						\shade[ball color=green] (7.2,-6.15) circle (2pt);
						\shade[ball color=green] (7.9, -6.75) circle (2pt);
						\shade[ball color=green] (9.8, -6.15) circle (2pt);
						\shade[ball color=green] (9.1, -6.75) circle (2pt);
						\shade[ball color=green] (10.1, -5.25) circle (2pt);
						\shade[ball color=green] (9.8, -4.3) circle (2pt);
						\shade[ball color=green] (9.1, -3.75) circle (2pt);
						\draw (7.7,-3.65) node {\small 1};
						\draw (7,-4.3) node {\small 5};
						\draw (6.7,-5.25) node {\small 2};
						\draw (7,-6.15) node {\small 10};
						\draw (7.7,-6.85) node {\small 6};
						\draw (9.3,-6.85) node {\small 3};
						\draw (10,-6.15) node {\small 7};
						\draw (10.3,-5.25) node {\small 4};
						\draw (10,-4.3) node {\small 8};
						\draw (9.3,-3.65) node {\small 9};
						\draw (13, 11.5) node{$G_3^{AB}=\emptyset$};
						\shade[ball color=green] (13, 10.5) circle (2pt);
						\draw (13, 10.2) node{4};
						\draw (12.5, 9) -- node[above] {\tiny $ A\xleftarrow{1} B$} node[below] {\tiny $ A\xrightarrow{2} B$} (13.5, 9);
						\shade[ball color=green] (12.5, 9) circle (2pt);
						\shade[ball color=green] (13.5, 9) circle (2pt);
						\draw (12.3, 8.8) node{4};
						\draw (13.7, 8.8) node{5};
						\draw (12.5, 6.75) -- (13.5, 6.75);
						\shade[ball color=green] (12.5, 6.75) circle (2pt);
						\shade[ball color=green] (13.5, 6.75) circle (2pt);
						\draw (12.3,6.55) node{4};
						\draw (13.7,6.55) node{5};
						\draw (12.5, 4.75) -- (13.5, 4.75);
						\draw (12.5, 4.75) -- node[below, sloped] {\tiny $ A\xrightarrow{2} B$} (13, 4);
						\draw (13, 4) -- node[below, sloped] {\tiny $ A\xleftarrow{2} B$} (13.5,4.75); 
						\shade[ball color=green] (12.5, 4.75) circle (2pt);
						\shade[ball color=green] (13.5, 4.75) circle (2pt);
						\shade[ball color=green] (13, 4) circle (2pt);
						\draw (12.3,4.95) node{4};
						\draw (13.7,4.95) node{5};
						\draw (13,3.65) node{7};
						\draw (12.5, 1.5) -- (13.5, 1.5);
						\draw (12.5, 1.5) -- (13, 0.75);
						\draw (13, 0.75) -- (13.5,1.5); 
						\shade[ball color=green] (12.5, 1.5) circle (2pt);
						\shade[ball color=green] (13.5, 1.5) circle (2pt);
						\shade[ball color=green] (13, 0.75) circle (2pt);
						\shade[ball color=green] (13, 2.25) circle (2pt);
						\draw (12.3,1.7) node{4};
						\draw (13.7,1.7) node{5};
						\draw (13,0.45) node{7};
						\draw (13,2.45) node{8};
						\draw (12.5, -1.75) -- (13.5, -1.75);
						\draw (12.5, -1.75) -- (13, -2.5);
						\draw (13, -2.5) -- (13.5,-1.75);
						\draw (12.5, -1.75) -- node[above, sloped] {\tiny $ A\xrightarrow{1} B$} (12.5,-1); 
						\draw (12.5, -1) -- node[above, sloped] {\tiny $ A\xleftarrow{2} B$} (13.5,-1);
						\shade[ball color=green] (12.5, -1.75) circle (2pt);
						\shade[ball color=green] (13.5, -1.75) circle (2pt);
						\shade[ball color=green] (13, -2.5) circle (2pt);
						\shade[ball color=green] (12.5, -1) circle (2pt);
						\shade[ball color=green] (13.5, -1) circle (2pt);
						\draw (12.3,-2) node{4};
						\draw (13.7,-2) node{5};
						\draw (13,-2.8) node{7};
						\draw (12.3,-0.7) node{9};
						\draw (13.7,-0.7) node{8};
						\draw (12.5, -5.75) -- (13.5, -5.75);
						\draw (12.5, -5.75) -- (13, -6.5);
						\draw (13, -6.5) -- (13.5,-5.75);
						\draw (12.5, -5.75) -- (12.5,-5); 
						\draw (12.5, -5) --(13.5,-5);
						\draw (12.5, -4.25) --node[above, sloped] {\tiny $ A\xrightarrow{2} B$} (13.5,-5);
						\shade[ball color=green] (12.5, -5.75) circle (2pt);
						\shade[ball color=green] (13.5, -5.75) circle (2pt);
						\shade[ball color=green] (13, -6.5) circle (2pt);
						\shade[ball color=green] (12.5, -5) circle (2pt);
						\shade[ball color=green] (13.5, -5) circle (2pt);
						\shade[ball color=green] (12.5, -4.25) circle (2pt);
						\draw (12.3,-6) node{4};
						\draw (13.7,-6) node{5};
						\draw (13,-6.8) node{7};
						\draw (12.3,-4.75) node{9};
						\draw (13.7,-4.75) node{8};
						\draw (12.2,-4) node{10};
			\end{tikzpicture}
		\end{center}
		\caption{Cycles $A_\sq$ and $B_\sq$ and corresponding $G_\sq^{AB}$ }
		\label{fig.1}
\end{figure}

  \begin{description}
     \item[$n=3$:] As mentioned above, $G^{AB}_{3}$ is a graph without vertices.

     \item[$n=4$:] Alice inserts her new node~$4$ between into the edge $\{1,2\}$ of her cycle~$A_3$; Bob inserts his new node~$4$ into the edge $\{1,3\}$ of his cycle~$B_3$.  Hence, $\{1,2\} = \nu_A(4) \ne \nu_B(4) = \{1,3\}$, so vertex~$4$ is added to $G_3^{AB}$.

     \item[$n=5$:] Alice inserts her new node~$5$ into the edge $\{2,4\}$ of her cycle~$A_4$; Bob inserts his new node~$5$ into the edge $\{1,2\}$ of his cycle~$B_4$.  Since $\{2,4\} = \nu_A(5) \ne \nu_B(5) = \{1,2\}$, vertex~$5$ is added to $G_4^{AB}$.  Let us check the edges:
     \begin{itemize}
     \item In $B_4$, the segment between $2$ and~$4$ contains the node~3 which is smaller than~$4$.  By Remark~\ref{rem:check-if-pair-is-a-nu}, there is no~$k$ with $\nu_A(5) = \nu_B(k)$, and hence no type-1 edge from $A$ to~$B$ at this time.
     \item As $\nu_B(5) = \nu_A(4)$, there is a type-1 edge between $4$ and~$5$ from $B$ to~$A$.
     \item Since $\max \nu_A(5)=4$ and $\nu_B(4) =\{1,3\} \not\ni 2$, there is also a type-2 edge between $5$ and~$4$ from $A$ to~$B$.
     \item MOZHGAN: Type-2 edge from $B$ to~$A$.
     \end{itemize}

     \item[$n=6$:] Alices inserts her new node~$6$ into the edge $\{2,3\}$ of her cycle, Bob inserts his new node~6 into the edge $\{2,3\}$ of his cycle.
       Since $\{2,3\} = \nu_A(6) = \nu_B(6) = \{2,3\}$, we don't have a vertex $6$ in $G_6^{AB}$.

     \item[$n=7$:] Alice throws into $\{4,5\}$, Bob throws into $\{3,4\}$.  Since $\{4,5\} = \nu_A(7) \ne \nu_B(7) = \{3,4\}$, the vertex~$7$ is added to $G_6^{A,B}$.
     \begin{itemize}
     \item MOZHGAN: Type-1 edge from $A$ to~$B$.
     \item MOZHGAN: Type-1 edge from $B$ to~$A$.
     \item As $\max \nu_A(7)=5$ and $\nu_B(5)=\{1,2\} \not\ni 4$, we have a type-2 edge from $A$ to~$B$ between $7$ and~$5$.
     \item As $\max \nu_B(7)=4$ and $\nu_A(4)=\{1,2\} \not\ni 3$, there is also a type-2 edge from $B$ to~$A$ between $7$ and~$4$.
     \end{itemize}

     \item[$n=8$:] Alice plays $\{3,6\}$, Bob chooses $\{1,4\}$.  Since $\{3,6\} = \nu_A(8) \ne \nu_B(8) = \{1,4\}$, the vertex~$8$ is added to $G_7^{A,B}$.
     \begin{itemize}
     \item In~$B_7$, the segment between $3$ and~$6$ is empty (just the edge).  By Remark~\ref{rem:check-if-pair-is-a-nu}, there is no~$k$ with $\nu_A(8) = \nu_B(k)$, and hence no type-1 edge from $A$ to~$B$.
     \item For the same reason (segment between $1$ and~$4$ empty), there is no type-1 edge from $B$ to~$A$ incident to the vertex~$8$.
     \item $\max \nu_A(8)=6$ and $\nu_B(6)=\{2,3\} \ni 3$.  So there is no type-2 edge from $A$ to $B$ between $8$ and a smaller vertex.
     \item $\max \nu_B(8)=4$ and $\nu_A(4)=\{1,2\} \ni 1$.  So there is no type-2 edge from $B$ to $A$ between $8$ and a smaller vertex.
     \end{itemize}
     Hence, vertex $8$~is isolated in $G_8$.

     \item[$n=9$:] Alice chooses $\{1,3\}$, Bob chooses $\{1,8\}$.  As $\{1,3\} = \nu_A(9) \ne \nu_B(9) = \{1,8\}$, the vertex~$9$ is added to $G_8^{A,B}$.
     \begin{itemize}
     \item As $\{1,3\} = \nu_A(9) = \nu_B(4) = \{1,3\}$, there is a type-1 edge from $A$ to~$B$  between $9$ and~$4$.
     \item MOZHGAN: Type-1 edge from $B$ to~$A$.
     \item MOZHGAN: Type-2 edge from $A$ to~$B$.
     \item As $\max \nu_B(9)=8$ and $\nu_A(8) = \{3,6\} \not\ni 1$, there is a type-2 edge from $B$ to~$A$ between $9$ and~$8$.
     \end{itemize}

     \item[$n=10$:] Alice chooses $\{3,9\}$, Bob chooses $\{2,6\}$.  Since $\{3,9\} = \nu_A(10) \ne \nu_B(10) = \{2,6\}$, the vertex $10$ is added to $G_9^{A,B}$.
     \begin{itemize}
     \item MOZHGAN: Type-1 edge from $A$ to~$B$.
     \item Again, in~$B$,  the segment between 3 and~9 has a vertex (4) smaller than~9.   Remark~\ref{rem:check-if-pair-is-a-nu} gives us that there is no~$k$ $\nu_B(k) = \nu_A(10)$,  so no type-1 edge from $B$ to~$A$ is created.
     \item As $\max \nu_A(10)=9$ and $\nu_B(9) = \{1,8\}\not\ni 3$, there is a type-2 edge from $A$ to~$B$ between $10$ and~$9$.
     \item As $\max \nu_B(10)=6$ and $\nu_A(6) = \{2,3\} \ni 2$, no type-2 edge from $B$ to~$A$ is created.
     \end{itemize}
  \end{description}
\end{example}

\subsection{Rephrasing Theorem \ref{thm:main-polytope}}
We now rephrase Theorem~\ref{thm:main-polytope}, in terms of the pedigree graph.  We also unravel the little-o, and move to the ``Alice-and-Bob'' letters for the cycles.

\begin{theorem}[Theorem~\ref{thm:main-polytope}, rephrased]\label{thm:main:pedigreegraph}
  For every $\eps >0$ there is an integer~$N$ such that for all $n\ge N$ and all cycles $A_n$ with node set $[n]$, if $B_n$ is drawn uniformly at random from all cycles with node set~$[n]$, then
  \begin{equation*}
    \Prb\bigl( \text{ $G^{AB}_n$ is connected } \bigr) \ge 1-\eps.
  \end{equation*}
\end{theorem}

In symbols, and using infinite cycles, this reads:
\begin{equation*}
  \forall \eps>0 \;\; \exists N\colon \; \forall A \, \forall n\ge N\colon \Prb( \text{ $G^{AB}_n$ is connected } ) \ge 1-\eps,
\end{equation*}
where the probability is taken over all infinite cycles, see the next section.  A close look at our proof shows that we are actually proving the following stronger statement (we don't have any use for it, though):
\begin{equation*}
  \forall \eps>0 \;\; \exists N\colon \; \forall A\colon \; \Prb\bigl( \text{ $\forall n\ge N$: $G^{AB}_n$ is connected } \bigr) \ge 1-\eps.
\end{equation*}

\section{Pedigree Graphs of Random Cycles}\label{Sec:rnd-cyc-basics}
We have to reconcile uniformly random cycles with the ``evolution over time'' concept of pedigrees.  The definition of an infinite cycle makes that very convenient, just do the same as with infinite sequences of coin tosses: Take, as probability measure on the sample space $\prod_{n=1}^\infty [n]$ of all infinite cycles the product of the uniform probability measures on each of the sets $[n]$, $n\ge 3$.  We refer to the atoms in this probability space as \textit{random infinite cycles.}  The following is a basic property of product probability spaces.  We will use it mostly without mentioning it.

\begin{fact}\label{fact:Bn-uniform}
  If $B$ is a random infinite cycle, then, for each $n\ge 3$, the cycle $B_n$ is uniformly random in the set of all cycles with node set $[n]$.
\end{fact}

\paragraph{Creating isolated vertices.} %
The first substantial result about the connectedness of the pedigree graph, concerns the creation of isolated vertices.

As outlined in the introduction, we study the situation in which Alice chooses her edge of~$A_{n-1}$ according to a sophisticated strategy, whereas Bob always chooses a uniformly random edge of $B_{n-1}$ to insert his node~$n$ into (which amounts to his cycle~$B_n$ being uniformly random in the set of all cycles on $[n]$, by Fact~\ref{fact:Bn-uniform}).  In this section, we adopt a purely ``random graph'' perspective.
For fixed~$A$ and random~$B$, the pedigree graphs $G^{AB}_{\sq}$ are a sequence of random graphs, with some weirdo distribution: At time~$n$, whether the new vertex~$n$ is added or not, and if it is, how many incident edges it has, and what their end vertices are --- these are all random events/variables.

For deterministic~$A$ and random~$B$, let the random variable~$Y$ count the total number of times that an isolated vertex of the pedigree graph is created.\label{txt:def-of-Y}  In other words, $Y = \sum_{n=4}^\infty \I{I_n}$, where $I_n$ denotes the event that, at time~$n$, $n$ is added as an isolated vertex to $G_n^{A,B}$ (and $\I{\sq}$ is the indicator random variable of the event).

\begin{lemma}\label{lem:isol-vert}
  Whatever Alice does, $\Exp Y = 2$.

  Moreover, for every $\eps>0$, if $n_0 \ge 4/\eps +2$, then, whatever Alice does
  \begin{equation*}
    \Prb\Bigl( \bigcup_{n \ge n_0} I_n \Bigr)  \quad \le \; \eps.
  \end{equation*}
\end{lemma}

For the proof we refer to the journal version of this extended abstract~\cite{Makkeh-Pourmoradnasseri-Theis:pedigree:18}.
To understand why the lemma is important, consider a pedigree graph at time $n$, just before Alice and Bob make their choices of cycle edges into which their respective new nodes~$n$ are inserted.  If~$n$ is not a vertex of the new pedigree graph~$G_n$, the number of connected components of~$G_{\sq}$ doesn't change.  If~$n$ is a vertex, and and it does have incident edges, then the number of connected components can only decrease.  The only way that the number of connected components of~$G_n$ can increase is if~$n$ is an isolated vertex in the new pedigree graph.  Hence, Lemma~\ref{lem:isol-vert} provides an upper bound on the expected number of connected components, uniform over~$n$.

\paragraph{The Intuition.}\label{paragraph:intuition} %
From Lemma~\ref{lem:isol-vert}, it is unlikely that the pedigree graph will have many components.  Indeed, intuitively, if only~2 isolated vertices are ever created, that means that most of the time either nothing happens (no new vertex) or edges are created, ultimately reducing the number of components, so the pedigree graph is connected.

While this basic intuition is essentially correct, a closer look reveals some subtleties.  First of all, Alice has a big sway in choosing the end vertices of new edges: she can pick the end vertices of type-2 edges from $A$ to~$B$; and she can influence the end vertices of type-1 edges (both directions).  

Secondly, Bob's choices are reduced by the low degrees of the vertices.  (A stronger version of~(a) is proved in~\cite{Makkeh-Pourmoradnasseri-Theis:pedigree:18}.)

\begin{lemma}
  The maximum degree of a vertex in a pedigree graph is at most~6:
  \begin{enumerate}[label=(\alph*),nosep]
  \item up to~2 to vertices created in the past; and
  \item up to~6 to future vertices.
  \end{enumerate}
\end{lemma}

Hence, if a vertex~$n_0$ was created as an isolated vertex or landed in a small connected component, Bob has only 4--6 shots at connecting it to another connected component.
The good news is that Alice can never ``shut down'' a connected component completely: Bob can always extend it by one more vertex.

\begin{lemma}\label{lem:connect-to-component}
  Let $C$ be a connected component of the pedigree graph $G^{AB}_{n-1}$.  There exists a $k\in C$ such that, no matter what Alice's move is at time~$n$, Bob has a move which creates the vertex~$n$ and makes it adjacent to~$k$.
\end{lemma}
\begin{proof}
  Take $k := \max V(C)$.
  Since $k$ is a vertex, we have $\abs{ \nu_B(k) \cap \nu_A(k) } \le 1$.  Suppose that $\nu^+_B(k) \notin \nu_A(k)$ (the other case is symmetric).  Then, the first time Bob inserts a node, say~$n'$, into the edge on the positive side of~$k$, this will create a type-2 edge ``from $B$ to~$A$'' between $n'$ and~$k$.  Since~$k$ is the newest vertex in its component, Bob has not yet inserted a node there, so he can insert~$n$ there, now.
\qed
\end{proof}

However, for Bob to make a disconnected pedigree graph connected, at some time, he will have to manage to insert his new node in such a way that it has two incident edges, linking two connected components at the same time.

There is no difficulty in realizing that Alice wouldn't stand a chance against a strategically playing Bob.   But we claim that the game between a clever Alice and a blindfolded Bob will turn in Bob's favour almost all of the time.

\smallskip\noindent%
Computer simulations give another indication that some care has to be taken implementing the basic intuition: Even for~$n$ as large as~100, even if Alice's cycle is chosen uniformly \textsl{at random} instead of adversarial, the frequency (in 100000 samples) with which we saw a connected pedigree graph was only about 84\%.  In the remaining 16\% of cases, the typical situation is that of one giant connected component containing almost every vertex, and one tiny component growing only very slowly.
%
This indicates that even a \textsl{disinterested} Alice can do some damage.

\section{The Connectivity Game}\label{Sec:cnct-game}
At each time, Alice moves first.  As already explained, she determines the cycle~$A$ by choosing, at each time~$n$, the edge of~$A_{n-1}$ into which her new node~$n$ will be inserted.  Then Bob moves.  He determines~$B$ in the same way, but (using Fact~\ref{fact:Bn-uniform}), he will draw the edge of $B_{n-1}$ into which his new node~$n$ is inserted uniformly at random from all edges of $B_{n-1}$, and his choice is independent of his earlier choices.

We say that Bob wins, if there exists an $n_0$ such that for all $n\ge n_0$, the pedigree graph $G^{AB}_n$ is connected.  We need Bob to win ``uniformly'', i.e., $n_0$ must not depend on Alice's moves.

Let the random variable $T_n$ denote the number of connected components in the pedigree graph $G^{AB}_n$.  To analyze the development of the random process $T_\sq$, it turns out to be useful to consider a second random process, $S_\sq$.  Denote by $\CmnE_n$ the set of edges that Alice's cycle and Bob's cycles have in common,
\begin{equation*}
  \CmnE_n := E(A_n) \cap E(B_n),
\end{equation*}
and let
\begin{equation*}
  S_n := \abs{\CmnE_n}
\end{equation*}
count the number of cycle edges that Alice and Bob have in common.
We will distinguish Alice's moves by whether or not she chooses a common cycle edge to place her new cycle node.  
The set $\CmnEa_n$ holds those common edges which are not incident on the edge which Alice chooses for her new node:
\begin{equation*}
  \CmnEa_n := \bigl\{ e \in \CmnE_n \mid e \cap \nu_A(n+1) = \emptyset \bigr\};
\end{equation*}
we let $\Sa$ count the edges in $\CmnEa$:
\begin{equation*}
  \Sa_n := \abs{\CmnEa_n}.
\end{equation*}
Finally, denote by $\RstE_n$ the set of edges in Bob's cycle which are neither common nor incident on Alice's chosen edge:
\begin{align*}
  \RstE_n   &:= \bigl\{ e \in E(B_n) \setminus \CmnE_n \mid e\cap \nu_A(n+1) = \emptyset \bigr\} \text{; and}\\
  \Rst_n    &:= \abs{ \RstE_n }.
\end{align*}

We are now ready to state and prove the transition probabilities.  They depend on whether Alice chooses, for her new node, a common edge --- we refer to that as a c-move by Alice --- or an edge which is in the difference $E(A_n)\setminus \CmnE_n$ --- we call that a d-move.

\begin{lemma}\label{lem:transition-prbs}
  The conditional probabilities
    \begin{equation*}
      \Prb\Bigl(  S_{n+1}=S_n + \Delta_S \;\; \land \;\; T_{n+1}=T_n + \Delta_T \;\; \mid \; B_n \Bigr)
    \end{equation*}
    satisfy these bounds (entries not shown are ``$=0$''):
  \begin{center}
    \renewcommand{\arraystretch}{1.333333}
    \begin{tabular}{c|c|c|c|c|l c c|c|c|c|c|l}
      \multicolumn{1}{c|}{$\Delta_T$}&\multicolumn{5}{c}{}                                            &             &\multicolumn{1}{c|}{$\Delta_T$}&\multicolumn{5}{c}{} \\[.5ex]
      \cline{2-5}\cline{9-12}
      $+1$       &$=\frac{\Sa_n}{n}$&$\le\frac{2}{n}$   &                &                &          &             &$+1$       &           &                  &                           &                &                   \\[.25ex]
      \cline{2-5}\cline{9-12}
      $0$        &                  &$=\frac{\Rst_n}{n}$&$\le\frac{2}{n}$&$=\frac{1}{n}$  &          &             &$0$        &           &$=\frac{\Sa_n}{n}$&$\le\frac{\Rst_n-T_n+1}{n}$&$\le\frac{4}{n}$&                   \\[.25ex]
      \cline{2-5}\cline{9-12}
      $-1$       &                  &                   &                &                &          &             &$-1$       &\hspace*{2em}&                 &\multicolumn{2}{|c|}{ $\ge\frac{T_n-1}{n}$ }    &                   \\[.25ex]
      \cline{1-6}\cline{8-13}
      {}         &\multicolumn{1}{c}{$-2$}&\multicolumn{1}{c}{$-1$}&\multicolumn{1}{c}{$0$}&\multicolumn{1}{c}{$+1$}&$\Delta_S$&&{}&\multicolumn{1}{c}{$-2$}&\multicolumn{1}{c}{$-1$}&\multicolumn{1}{c}{$0$}&\multicolumn{1}{c}{$+1$}& $\Delta_S$        \\
      \multicolumn{6}{c}{c-move}                                                                     &\hspace{2em} &\multicolumn{6}{c}{d-move} \\[-1ex]
      \multicolumn{6}{c}{\footnotesize(Alice chooses common edge)}                                   &             &\multicolumn{6}{c}{\footnotesize(Alice chooses edge in $E(A_n)\setminus \CmnE_n$)}\\
    \end{tabular}
  \end{center}
\end{lemma}
The proof of this lemma requires some delicated distinguishing of cases, we refer to the journal version~\cite{Makkeh-Pourmoradnasseri-Theis:pedigree:18}.

The proof of the main theorem now follows the following idea.  From the tables in Lemma~\ref{lem:transition-prbs}, you see that d-moves have chance of reducing the number of connected components --- albeit a small one.  Moreover, Alice cannot take a c-move only when $S_n > 0$, but c-moves have a strong tendency to reduce $S_\sq$.  We prove that the number of d-moves that Alice has to take are frequent enough to lead to a decrease in the number of connected components.  This suffices to prove Theorem~\ref{Sec:mainproof}, along the lines sketched on page~\pageref{paragraph:intuition}.

The next section gives more details of the proof.

\section{Proof of Theorem~\ref{thm:main:pedigreegraph}}\label{Sec:mainproof}
Using the Azuma-Hoeffding super-martingale tail bound, we can prove that, for large enough $n_0$, Alice has to take many d-moves between times $n_0$ and $2n_0$.  Due to space restrictions, we have to refer to the journal version~\cite{Makkeh-Pourmoradnasseri-Theis:pedigree:18} for all of the proofs.

\begin{lemma}\label{lem:Alice-will-play--d--often}
  For every $\eps\in\lt]0,1\rt[$, if $n_0 \ge \max(900, 8\ln(1/\eps))$, and $n_1 := 2 n_0$ then, whatever Alice does, the probability, conditioned on $B_{n_0}$ and $S_{n_0}\le \ln^2 n_0$, that among her moves at times $n = n_0+1,\dots,n_1$, there are fewer than $n_0/3$ d-moves, is at most~$\eps$.
\end{lemma}

From this, we deduce that must $T_\sq$ decrease, but some sophistication is needed, because of the slow divergence of $\sum 1/n$: Indeed, between $n_0$ and $2n_0$, $T_\sq$ decreases only with a constant probability:

\begin{lemma}\label{lem:T-will-decrease}
  Fix $\delta := \nfrac{1}{42}$.  If $n_0 \ge \max(900, 8\ln(1/\delta))$, and $n_1 := 2 n_0$ then, whatever Alice does,
  \begin{equation*}
    \Prb\Bigl( \exists n \in \{n_0+1,\dots,n_1\}\colon \;\; T_{n+1} < T_n \;\; \Bigm| B_{n_0}, T_{n_0} \ge 2, S_{n_0} \le \ln^2 n_0 \Bigr)  \quad \ge \quad 1/7.
  \end{equation*}
\end{lemma}

We can boost the probability to $1-\eps$, for arbitrary $\eps >0$, by iterating the argument $\Omega(\ln(1/\eps))$ times.

\begin{lemma}\label{lem:T-will-decrease:boost}
  Fix $\delta := \nfrac{1}{42}$.  For every $\eps\in\lt]0,\nfrac{1}{56}\rt[$, with $a := 10\ln(2/\eps)$, if
  \begin{equation*}
      n_0 \ge \max(900, 8\ln(1/\delta), (2a)^{4/\eps}, e^{6/\eps}),
  \end{equation*}
  and $n_1 := 2 a n_0$ then, whatever Alice does,
  \begin{equation*}
    \Prb\Bigl( \exists n \in \{n_0,\dots,n_1\}\colon \;\; T_{n+1} < T_n \;\; \Bigm| B_{n_0}, T_{n_0} \ge 2 \Bigr)  \quad \ge \quad 1-\eps
  \end{equation*}
\end{lemma}

Note that Lemma~\ref{lem:T-will-decrease:boost} also gets rid of the conditioning on $S_n \le \ln^2 n$.

We are now ready to complete the proof of the main theorem.

\begin{proof}[of Theorem~\ref{thm:main:pedigreegraph}.]
  Let $\eps' \in\lt]0,\nfrac{1}{2}\rt[$ be given.  Set $t := 6/\eps'$.  Since $T_{\sq}$ can only increase when an isolated vertex is created, we have $T_n \le Y$, for all $n\ge 4$, where~$Y$ is the number of isolated vertices.  Hence, by Lemma~\ref{lem:isol-vert} and Markov's inequality, we have
  \begin{equation*}
    \Prb\Bigl( \exists n \ge 4\colon \;\; T_n \ge t+1 \Bigr)
    \le
    \Prb( Y \ge t )
    \le
    \Exp(Y) / t
    = \eps'/3.
  \end{equation*}

  Now take $n_0' \ge 12/\eps' +2$, and large enough to apply Lemma~\ref{lem:T-will-decrease:boost} $n_0:=n_0'$ and to $\eps := \eps'/3t$ (note that this is less than $1/56$).  Denote by $a$ be the number defined in that lemma.  Applying the lemma~$t$ times, for $n_0$ ranging over $n'_0 +j2a n'_0$, $j=0,\dots,t-1$, the probability that we fail at least once to obtain a decrease in the number of connected components, $T_\sq$, is at most $\eps'/3$.  So, with probability at least $1-2\eps'/3$, we must have $T_{n_0}=1$ for one of these $n_0$'s or for an $n$ between $n'_0 +(t-1)2a n'_0$ and $n'_0 +t2a n'_0$.

  Finally, since $n_0' \ge 12/\eps' +2$, by Lemma~\ref{lem:isol-vert}, with probability $1-\eps'/3$, $T_\sq$ will not increase after $n'_0$, and hence, with probability $1-\eps'$, will drop to~1 and stay there for all eternity.  Bob wins.
\end{proof}

\section{Some Open Questions}
There are two questions which we believe should be asked in the context of our result.

Firstly, are there other polytopes whose graphs are not complete, but the minimum degree is asymptotically that of a complete graph?  Could that even be the case for the Traveling Salesman Problem polytope itself?

Secondly, in view of the Traveling Salesman Problem polytope, it would be interesting to find other combinatorial conditions on cycles which are implied by the adjacency of the corresponding vertices on the TSP polytope.  The pedigree graph connectedness condition is derived from an extension of the TSP polytope, but maybe there are other combinatorial conditions without that geometric context.  The graph resulting from such a condition might be ``closer'' to the actual TSP polytope graph, i.e., add fewer edges.



\begin{thebibliography}{10}
\providecommand{\url}[1]{\texttt{#1}}
\providecommand{\urlprefix}{URL }

\bibitem{Aguilera:adj-setcover-polyh:2014}
Aguilera, N., Katz, R., Tolomei, P.: Vertex adjacencies in the set covering
  polyhedron. arXiv preprint arXiv:1406.6015  (2014)

\bibitem{Arthanari:DiscreteMath:06}
Arthanari, T.S.: On pedigree polytopes and hamiltonian cycles. Discrete Math.
  306,  1474--1792 (2006)

\bibitem{Arthanari:DiscreteOpt:13}
Arthanari, T.S.: Study of the pedigree polytope and a sufficiency condition for
  nonadjacency in the tour polytope. Discrete Optimization  10(3),  224--232
  (2013), \url{http://dx.doi.org/10.1016/j.disopt.2013.07.001}

\bibitem{Arthanari:ped-one:2000}
Arthanari, T.S., Usha, M.: An alternate formulation of the symmetric traveling
  salesman problem and its properties. Discrete Applied Mathematics  98(3),
  173--190 (2000)

\bibitem{Fiorini-Kaibel-Pashkovich-Theis:CombLB:13}
Fiorini, S., Kaibel, V., Pashkovich, K., Theis, D.O.: Combinatorial bounds on
  nonnegative rank and extended formulations. Discrete Math.  313(1),  67--83
  (2013)

\bibitem{Groetschel-Padberg:polyh-th:1985}
Gr{\"o}tschel, M., Padberg, M.W.: Polyhedral theory. In: Lawler, E.L., Lenstra,
  J.K., Kan, A.H.G.R., Shmoys, D.B. (eds.) The Traveling Salesman Problem. A
  Guided Tour of Combinatorial Optimization, chap.~8, pp. 251--306. Wiley
  (1985)

\bibitem{Kaibel:lowdim-faces-rndptp:2004}
Kaibel, V.: Low-dimensional faces of random 0/1-polytopes. In: International
  Conference on Integer Programming and Combinatorial Optimization. pp.
  401--415. Springer (2004)

\bibitem{Kaibel-Pashkovich-Theis:SIDMA:2012}
Kaibel, V., Pashkovich, K., Theis, D.O.: Symmetry matters for sizes of extended
  formulations. SIAM Journal on Discrete Mathematics  26(3),  1361--1382 (2012)

\bibitem{Kaibel-Remshagen:03}
Kaibel, V., Remshagen, A.: On the graph-density of random 0/1-polytopes. In:
  Approximation, Randomization, and Combinatorial Optimization: Algorithms and
  Techniques, 6th International Workshop on Approximation Algorithms for
  Combinatorial Optimization Problems, {APPROX} 2003 and 7th International
  Workshop on Randomization and Approximation Techniques in Computer Science,
  {RANDOM} 2003, Princeton, NJ, USA, August 24--26, 2003, Proceedings. pp.
  318--328 (2003), \url{http://dx.doi.org/10.1007/978-3-540-45198-3_27}

\bibitem{Makkeh-Pourmoradnasseri-Theis:pedigree:18}
Makkeh, A., Pourmoradnasseri, M., Theis, D.O.: On the graph of the pedigree
  polytope (2016), arXiv:1611.08431

\bibitem{Maksimenko:adj-np:2014}
Maksimenko, A.: The common face of some 0/1-polytopes with np-complete
  nonadjacency relation. Journal of Mathematical Sciences  203(6),  823--832
  (2014)

\bibitem{Naddef:JCTB:84}
Naddef, D.: Pancyclic properties of the graph of some 0--1 polyhedra. Journal
  of Combinatorial Theory, Series B  37(1),  10--26 (1984)

\bibitem{NaddefRinaldi93}
Naddef, D., Rinaldi, G.: The graphical relaxation: a new framework for the
  symmetric traveling salesman polytope. Math. Programming  58(1, Ser. A),
  53--88 (1993), \url{http://dx.doi.org/10.1007/BF01581259}

\bibitem{Naddef-Pulleyblank:JCTB:84}
Naddef, D.J., Pulleyblank, W.R.: Hamiltonicity in (0--1)-polyhedra. Journal of
  Combinatorial Theory, Series B  37(1),  41--52 (1984)

\bibitem{OswaldReineltTheis07}
Oswald, M., Reinelt, G., Theis, D.O.: On the graphical relaxation of the
  symmetric traveling salesman polytope. Math. Program.  110(1, Ser. B),
  175--193 (2007), \url{http://dx.doi.org/10.1007/s10107-006-0060-x}

\bibitem{Papadimitriou:adj-tsp:78}
Papadimitriou, C.H.: The adjacency relation on the traveling salesman polytope
  is np-complete. Mathematical Programming  14(1),  312--324 (1978)

\bibitem{Pashkovich:hidden:15}
Pashkovich, K., Weltge, S.: Hidden vertices in extensions of polytopes.
  Operations Research Letters  43(2),  161--164 (2015)

\bibitem{Santos:hirsch:2012}
Santos, F.: A counterexample to the hirsch conjecture. Annals of mathematics
  176(1),  383--412 (2012)

\bibitem{Sarangarajan:tspskel-density-LB:1997}
Sarangarajan, A.: A lower bound for adjacencies on the traveling salesman
  polytope. SIAM Journal on Discrete Mathematics  10(3),  431--435 (1997)

\bibitem{Schrijver:Book:03}
Schrijver, A.: Combinatorial optimization. {P}olyhedra and efficiency.,
  Algorithms and Combinatorics, vol.~24. Springer-Verlag, Berlin (2003)

\bibitem{Sierksma:skeleton:1993}
Sierksma, G.: The skeleton of the symmetric traveling salesman polytope.
  Discrete applied mathematics  43(1),  63--74 (1993)

\bibitem{Sierksma:DAM:00}
Sierksma, G., Teunter, R.H.: Partial monotonizations of hamiltonian cycle
  polytopes: dimensions and diameters. Discrete applied mathematics  105(1),
  173--182 (2000)

\bibitem{TheisDAM10}
Theis, D.O.: A note on the relationship between the graphical traveling
  salesman polyhedron, the symmetric traveling salesman polytope, and the
  metric cone. Discrete Appl. Math.  158(10),  1118--1120 (2010),
  \url{http://dx.doi.org/10.1016/j.dam.2010.03.003}

\bibitem{Theis:DiscreteOpt:2014}
Theis, D.O.: On the facial structure of symmetric and graphical traveling
  salesman polyhedra. Discrete Optimization  12,  10--25 (2014),
  \url{http://www.sciencedirect.com/science/article/pii/S1572528613000625}

\end{thebibliography}
\end{document}